\documentclass{article}

     \PassOptionsToPackage{numbers, compress}{natbib}
  \usepackage[preprint]{neurips_2026}

\usepackage[utf8]{inputenc} 
\usepackage[T1]{fontenc}    
\usepackage{hyperref}       
\usepackage{url}            
\usepackage{booktabs}       
\usepackage{amsfonts}       
\usepackage{nicefrac}       
\usepackage{microtype}      
\usepackage[table]{xcolor}         
\usepackage{amsmath,graphicx}
\usepackage{hyperref}
\usepackage{amsmath}
\usepackage{times}
\usepackage{subcaption}
\usepackage{bm}
\usepackage{amsthm}
\usepackage{algorithm}
\usepackage{algpseudocode}
\usepackage{float}
\usepackage{amssymb,amsfonts,graphicx,float,multirow,enumitem, url,booktabs}
\usepackage{bbm}
\usepackage{comment}
\usepackage{wrapfig}
\usepackage{placeins}
\usepackage{graphicx,caption} 
\usepackage{blindtext} 
\usepackage{adjustbox} 

\newtheorem{theorem}{Theorem}

\newtheorem{assumption}{Assumption}
\newtheorem{lemma}{Lemma}
\newtheorem{remark}{Remark}

\def\E{{\textsf{E}}}
\def\Var{{\textsf{var}}}
\def\Cov{{\textsf{cov}}}

\allowdisplaybreaks

\usepackage[normalem]{ulem}

\newcommand{\pr}{\textsf{pr}} 
\newcommand{\ep}{\textsf{E}} 
\newcommand{\var}{\textsf{var}} 
\newcommand{\cov}{\textsf{cov}} 
\newcommand{\bad}[1]{\cellcolor{red!12}\text{#1}}

\graphicspath{{./../images/Performance Comparison/curves_3/}}

\title{KAP-CPD: Kernel Aggregation for Change-Point Detection in Dynamic Networks}

\author{%
  Mingxuan Sun  and Hao Chen 
  \\
  Department of Statistics \\
  University of California, Davis\\
  Davis, CA 95616 \\
  \texttt{mxsun@ucdavis.edu; hxchen@ucdavis.edu} \\
}

\begin{document}

\maketitle

\begin{abstract}
  Change-point detection in dynamic networks has received much attention due to its broad applications in social networks and biological systems. Kernel-based methods have shown strong potential for this problem. However, their performance can depend sensitively on the choice of kernel, and selecting an appropriate kernel is challenging when the underlying change pattern is unknown. Motivated by this challenge, we propose KAP-CPD, a new kernel-based testing framework for change-point detection in dynamic networks. KAP-CPD aggregates information from multiple kernels, allowing it to adapt to diverse change patterns. The proposed method does not assume specific underlying network distribution, and achieves strong empirical power across a wide range of network change scenarios. To improve scalability, we further develop a fast analytic testing procedure, KAPf-CPD, that substantially reduces computation time for long network sequences compared with permutation-based alternatives and current state-of-the-art methods. We evaluate our proposed framework through extensive simulations and real-world data on email communication networks and brain functional connectivity networks.

\end{abstract}

\section{Introduction}
    Network data are increasingly prevalent, as it is convenient to model the dynamics of interactions using network structure. For instance, social dynamic can be modeled with a network structure where the set of vertices denote individuals and the set of edges indicate pairwise interactions. In neuroscience, neural networks are represented as dynamic graphs with nodes as brain regions and edges as functional connections from fMRI data \cite{dynamic_graph_survey}. For example, epilepsy is a form of brain network disorder, and change-point detection methods that aim to identify physiological changes in these time-evolving epilepsy networks can contribute to epilepsy diagnosis, treatment and prognostics \cite{royer_epilepsy_2022, epilepsy_time-evolving_2024, NCPD}. In biology, it is also hypothesized that dynamic changes between genes are related to the clinical response, and we can model the changes using gene co-expression networks and apply change point detection methods to detect the changes \cite{gene_coexpression_cpd}. In these applications, change-point detection in dynamic graphs is essential for extracting insights from evolving modeled by networks of varying lengths from  dozens \cite{gene_coexpression_cpd,NCPD} to thousands \cite{seizure-detection}. In this paper, we consider the following offline change-point detection problem for dynamic networks: Given a sequence of independent observations $\{G_{t}\}_{t=1,\ldots,n}$, where each $G_{t}$ represents a snapshot of the network at a time point $t$, and each $G_{t}$ consists of a set of vertices $V_t$ and a set of edges $E_t$. We consider testing the null hypothesis:
    \begin{equation} \label{eq:H0}
    	H_0: \, G_t \sim F_0, \, t = 1, \ldots, n 
    \end{equation} 
    against the single change-point alternative:
    \begin{equation}  
    \label{eq:H1}
    	H_1: \exists \, 1 \le \tau < n, \, \, G_t \sim 
    	\begin{cases}
    		F_0,  \,\,\,\,  t \le \tau \\
    		F_1,  \,\,\,\, \text{otherwise}
    	\end{cases}
    \end{equation}
    where $F_{0}$ and $F_{1}$ are two different distributions.

\subsection{Related Work}
    \textbf{Parametric Models.}
        Several parametric approaches have been developed specifically for network-valued data. These methods typically assume a particular network-generating mechanism, such as the stochastic block model (SBM) \cite{JMLR:SBM}, more general Bernoulli random graph models \cite{wang2020optimalchangepointdetection}, preferential attachment models \cite{cpd_preferential_attachment,preferential-attachment}, or separable temporal exponential random graph models (STERGMs) \cite{kei2025cpdseparable}. Some of these works establish theoretical guarantees, including asymptotic consistency and minimax optimality \cite{wang2020optimalchangepointdetection}. However, parametric approaches often rely on restrictive model assumptions that may be violated in real applications, which can substantially degrade their empirical performance. 
        
    \textbf{Nonparametric Models.}
        There are some methods based on matrix factorization and utilize spectral information on the adjacency or the Laplacian matrices of networks. For example, in \cite{LAD}, the authors use singular value decomposition of the graph Laplacian as the low-dimensional representation. In \cite{NCPD} the authors use spectral clustering on the Laplacian matrix. There are other nonparametric models that are not designed particularly for dynamic networks but can be generally applied to network-structured data. For instance, \cite{kerseg} proposed a framework with single kernel based statistic. There is also Fréchet statistics based method \cite{frechet-cpd}, rank-based method \cite{rank-based-cp} and graph-based methods \cite{chen2015graph,chu-and-chen19,song-chen-2022}.  
        
    \textbf{Deep Learning.}
     Authors in \cite{graph_similarity_learning} adopts a graph neural network architecture to learn a data-driven graph similarity function from a training subsequence. However, reliance on training data is perhaps unrealistic in real-world change-point detection since problems are typically unsupervised. Unsupervised deep learning approaches have been proposed to reduce this dependence. For instance, authors in \cite{deep-learning-NCPD-GMM} proposes VGGM, which jointly trains a variational graph autoencoder (VGAE) and a Gaussian mixture model (GMM), while authors in \cite{kei2025-decoder} detects change points in low-dimensional latent representations using a decoder-only architecture. Nevertheless, these methods may still depend on assumptions about the underlying network distribution or latent representation, and their detection thresholds can be difficult to calibrate in practice.
        
\subsection{Motivation: Difficulty of Kernel Selection}
    Kernels are widely used in two-sample testing \cite{JMLR:gretton-MMD}, and kernel-based change-point detection methods have been shown to perform well across a broad range of alternatives and data types \cite{JMLR:kernel-change-point}. A key challenge, however, is that kernel methods require selecting a kernel \emph{a priori}. Different kernels have strengths and weaknesses in different scenarios, a poorly chosen kernel can lead to substantial power loss. This motivates the development of kernel-combination strategies that can leverage complementary information from multiple kernels. Kernel combination has been studied in two-sample testing \cite{2023Boosting-kernel-two-sample-test,2023mmdfuse,zhou2025dual-kernel-combination,JMLR:combining-kernels}. Motivated by these developments, we investigate kernel combination strategies in change-point detection in dynamic network setting. In practice, Gaussian RBF kernels are commonly used. However, for network-valued data, Gaussian kernels require vectorizing graphs into a Euclidean space. While convenient, such representations may fail to capture higher-order network structure, such as transitivity which are often central to network dynamics \cite{network_representation_learning_survey}. As a result, changes that primarily affect higher-order dependencies may be difficult to detect using Gaussian kernels alone. In contrast, graph kernels operate directly on graph objects and can capture structural information such as subgraphs \cite{pmlr-v5-graphlet-kernel}, paths \cite{random-walk-kernel,graph-kernels-general-framework} and isomorphism \cite{WLkernel}. Combining diverse kernels therefore enables us to obtain the best of both worlds. Proper aggregation of these kernels should preserve the strengths of each as much as possible while making up their individual limitations, thereby improving detection power across a broader class of change-point alternatives.

    \subsection{Type I Error Control and Inference}
    In addition to kernel selection, inference for network change-point detection poses nontrivial challenges. Many existing methods establishes type-I error control through a model-based data-driven threshold \cite{wang2020optimalchangepointdetection,kei2025-decoder,kei2025cpdseparable}. However, when the underlying data distribution deviates from model assumptions, such procedures often struggle to reliably control type-I error. (See ~\ref{experiments: empirical-size}). Moreover, the detection threshold can also be computationally expensive to tune. Permutation-based tests provide a flexible alternative. By approximating the null distribution through data-driven resampling, permutation procedures provides natural and reliable type-I error control in finite-sample settings under exchangeability assumption. When combined with kernel aggregation, permutation-based inference further enhances the practical applicability of the proposed methodology by enabling valid hypothesis testing without relying on restrictive model assumptions. 


\section{New Test Procedure: KAP-CPD} \label{sec:new}
    \subsection{Definitions}\label{subsec:definition}
    Throughout the paper, we work under the permutation null distribution, which assigns probability $1/n!$ to each of the $n!$ permutations of the network sequence $\{G_t\}_{t=1}^n$. We use $\pr$, $\ep$, $\var$, and $\cov$ to denote probability, expectation, variance, and covariance under permutation null distribution. For an arbitrary quantity $T$, throughout the paper we let $Z_T$ denote the standardized version of $T$, where standardization is given by:$ Z_T=\frac{T-\ep(T)}{\sqrt{\Var(T)}}$.

    Let $K_1$ and $K_2$ be two kernel matrices computed from the network sequence $\{G_t\}_{t=1}^n$. Generally, the proposed framework can accommodate arbitrary user-selected kernels $K_1$ and $K_2$. Let $k_{xij}$ denote the $(i,j)$th entry of $K_x$ for $x \in \{1,2\}$, and let $\mathbbm{1}(\cdot)$ denote the indicator function.
        {\small
    \begin{equation*}
    \alpha_x(t)
    = \frac{1}{t(t-1)}
       \sum_{i=1}\sum_{j\ne i}
       k_{xij}\,\mathbbm{1}\{i,j\le t\}, \quad 
    \beta_x(t)
    = \frac{1}{(n-t)(n-t-1)}
       \sum_{i=1}\sum_{j\ne i}
       k_{xij}\,\mathbbm{1}\{i,j>t\}.
    \end{equation*}
    }
    Let $\vec a (t) =[\alpha_1(t), \beta_1(t),\alpha_2(t),\beta_2(t)]^T$, $\Sigma(t) = \var(\vec a (t))$:
    \begin{equation}
    \text{S}(t) =  
        [\vec a (t)- \E\vec a (t)]^{\mathsf{T}} 
             \Sigma(t)^{-1} 
             [\vec a (t)-\E\vec a (t)].
    \end{equation}\label{statistic-eq}
    Each element of $\Sigma(t)$ can be calculated and the specific expressions are given in Lemma \ref{lemma:first-two-moments}.
    \begin{remark}
    We note that \text{S}(t) can also incorporate more than 2 kernels by setting $$\vec a (t) =[\alpha_1(t),\beta_1(t)\dots \alpha_m(t),\beta_m(t)]^T.$$ In this paper, we focus on the case with two kernels as they already provide an effective combination that covers a broad range of alternatives as seen in our numerical studies in Section ~\ref{sec:performance}. 
    \end{remark}
    
    To test $H_0$ (\ref{eq:H0}) vs. $H_1$ (\ref{eq:H1}), we use the scan statistic:
    \begin{equation}\label{scan_stat}
        S^*=\max_{n_0\le t \le n_1} \text{S}(t)
    \end{equation}
    where $n_0$ and $n_1$ are pre-specified cut-offs for potential change-point location, the default setting is $n_0=0.05n$, $n_1=0.95n$. We reject $H_0$ when     Eq~\ref{scan_stat} is larger than the critical value for
    a given significance level. The estimated location of the change-point $\tau$ is given by: $\hat \tau = \arg\max_{n_0\le t \le n_1} \text{S}(t).\label{stat:tauhat}$
   

    \subsection{Computation Cost Reduction: Decomposition of S(t)}
        To further analyze the properties of the new statistic $\text{S}(t)$, we show that it could be decomposed into 4 separate components consists of the following fundamental quantities. Let $x\in \{1,2\}$. Define
        $$
            W_x(t) =\frac{t}{n}\alpha_x(t)+\frac{n-t}{n}\beta_x(t), \quad D_x(t) = \frac{t(t-1)}{n(n-1)}\alpha_x(t)-\frac{(n-t)(n-t-1)}{n(n-1)}\beta_x(t).\\
        $$     
        \begin{theorem}\label{thm:decomposition}
            $\text{S}(t)$ can be decomposed as follows:
            $$\text{S}(t)= Z_{W_{\text{Diff}}}^2(t)+Z_{D_{\text{Diff}}}^2(t)+Z_{W_{\text{Sum}}}^2(t)+Z_{D_{\text{Sum}}}^2(t).$$
            Where $Z_{W_{\text{Diff}}}(t)$, $Z_{D_{\text{Diff}}}(t)$, $Z_{W_{\text{Sum}}}(t),$ $Z_{D_{\text{Sum}}}(t)$ 
            are uncorrelated, and are the corresponding standardized version of:
            \[
                W_{\text{Diff}}(t) = W_1(t) - W_2(t), \quad D_{\text{Diff}}(t) = D_1(t) - D_2(t),
            \]
            \[
                W_{\text{Sum}}(t) = c_1(t)W_1(t) + W_2(t), \quad D_{\text{Sum}}(t) = c_2(t) D_1(t) + D_2(t).
            \]
            where $c_1(t)$ and $c_2(t)$ are two weights depending on values of $K_1$ and $K_2$ with detailed form given in Appendix $\ref{appendix:proof_thm1}$.
        \end{theorem}
        As shown in Theorem~\ref{thm:decomposition}, the quantities $W_1(t)$, $W_2(t)$, $D_1(t)$, and $D_2(t)$ serve as the fundamental components of the test statistic $\text{S}(t)$. Consequently, both the construction and the existence of $\text{S}(t)$ rely critically on these underlying terms. From a computational perspective, the decomposition in Theorem~\ref{thm:decomposition} also leads to a significant reduction in computational cost, as it eliminates the need to invert a $4 \times 4$ matrix at every time step $t$. 
        
        \begin{theorem}\label{thm2}
            For $2 \le t \le n-2$, $S(t)$ is well-defined except when either $\Var(W_1(t))\Var(W_2(t))=\Cov(W_1(t),W_2(t))^2 \text{ or } \Var(D_1(t))\Var(D_2(t))=\Cov(D_1(t),D_2(t))^2$ , or when $\Var(D_1(t)+D_2(t))=0$ or $\var(W_1(t)+W_2(t))=0$.
        \end{theorem}
        \begin{remark}
         These conditions pertain to the invertibility of $\Sigma(t)$. In practice, combining kernels that are perfectly linearly correlated (or highly correlated) would be undesirable as they contain redundant information.
        \end{remark}

    \subsection{Overall Testing Procedure for KAP-CPD}

    Given the scan statistic $\text{S}(t)$, we test for the presence of a change point by evaluating the tail probability of the scan statistic \ref{scan_stat} under the null hypothesis $H_0$  \eqref{eq:H0}:
    \begin{equation} 
    \label{tail:single}
        \pr \left(
        \max_{n_0 \le t \le n_1} S(t) > b
        \right).
    \end{equation}
    
    To estimate \eqref{tail:single}, we use a permutation procedure which is valid under the  exchangeability condition:
    \begin{assumption}
    \label{assumption:exchangeability}
    Under $H_0$, the network sequence $\{G_t\}_{t=1}^n$ is exchangeable. 
    \end{assumption}
    In particular, this holds if $G_1,\ldots,G_n$ are independent and identically distributed under $H_0$.

    \begin{algorithm}[H]
        \caption{Permutation Testing Procedure for KAP-CPD}
        \label{alg:fusion_cpd}
        \begin{algorithmic}[1]
        \Require Kernel Matrices $K_1,K_2$ computed from $\mathcal{X}_n=\{G_t\}_{t=1}^n$, number of permutations $B$, $n_0,n_1$.
        \State Compute the observed scan statistic $S^* = \max_{n_0\le t\le n_1} S(t).$

        \For{$b=1,\ldots,B$}
            \State Generate a random permutation $\pi_b$ of $\{1,\ldots,n\}$.
            \State Permute the network sequence:
            $
            \mathcal{X}_n^{\pi_b}
            =
            \{G_{\pi_b(1)},\ldots,G_{\pi_b(n)}\}.
            $
            \State Compute the permuted scan statistic
            $
            S_b^\pi = \max_{n_0\le t\le n_1} S_b^\pi(t).
            $
        \EndFor
        
        \State Compute the permutation $p$-value:
        $
        p_{\text{perm}}
        =
        \frac{1}{B}
        \sum_{b=1}^B
        \mathbbm{1}\left\{S_b^\pi \ge S^*\right\}.
        $
        
        \State \Return $p_{\text{perm}}$, $\hat \tau$.
        \end{algorithmic}
        \end{algorithm}

\section{Analytical p-value approximations and fast test} \label{sec:theory}
    While the $p$-value in \eqref{tail:single} can be computed using permutation methods, such procedures are computationally expensive. In this section, we study the asymptotic behavior of the fundamental components. By leveraging their asymptotic tail probabilities, we can obtain analytical approximations to the corresponding 
    $p$-values, resulting in a fast testing procedure that is well suited for the initial screening of potential change-points. We first mention that $W_x(t)$ is equivalent to MMD$(t)$ up to a constant \cite{Generalized_ker_two_sample}, and therefore it is degenerate and does not converge to gaussian under the null hypothesis \cite{JMLR:gretton-MMD}. Following the procedure in \cite{Generalized_ker_two_sample}, for the asymptotic distribution, we instead work on a $r$-weighted version of $W_x(t)$, where $r$ is a constant and $x \in \{1,2\}$:
    $
        W_{x,r} (t) = r \frac{t}{n} \alpha_x(t) + \frac{n-t}{n} \beta_x (t).
    $

    \subsection{Fast test: KAPf-CPD}
\label{subsec:fast}

    We develop a fast version of KAP-CPD, denoted KAPf-CPD, by using a separate Bonferroni-correction procedure to provide analytical $p$-value in place of permutation $p$-value. The procedure uses the limiting distributions of the standardized processes
    $Z_{D_x}(t)$ and $Z_{W_{x,r}}(t)$, discussed in Section~\ref{fast-test:limiting-dist-p} and Appendix~\ref{appendix:limiting-dist-details}. Define $Z_{D_x}^*=\max_{n_0\le t\le n_1} |Z_{D_x}(t)|$, $Z_{W_{r,x}}^*=\max_{n_0\le t\le n_1} Z_{W_{r,x}}(t).$
    We compute the approximate Bonferroni-adjusted $p$-value:
    \begin{equation}
    \label{eq:fast-test}
    p_{\mathrm{approx}}
    =
    \min\left\{
    1,\,
    6p_{Z^*_{D_1}},
    6p_{Z^*_{D_2}},
    6p_{Z^*_{W_{1,r_1}}},
    6p_{Z^*_{W_{1,r_2}}},
    6p_{Z^*_{W_{2,r_1}}},
    6p_{Z^*_{W_{2,r_2}}}
    \right\},
    \end{equation}
    where each term denotes the corresponding $p$-value approximations. KAPf-CPD rejects $H_0$ when $p_{\mathrm{approx}}$ is below the nominal significance level. If $H_0$ is rejected, the change-point location is estimated with $\arg\max_{n_0 \le t\le n_1} \text{S}(t)$ in \eqref{stat:tauhat}. Thus, \eqref{eq:fast-test} provides a computationally efficient alternative to the permutation $p$-value $p_{\mathrm{perm}}$ in Algorithm~\ref{alg:fusion_cpd}.

    \subsection{Asymptotic p-value approximations} \label{fast-test:limiting-dist-p}
        KAPf-CPD relies on the asymptotic behavior of the standardized processes:
        \begin{equation}\label{eq:W-D-processes}
        \{Z_{D_x}(\lfloor nu \rfloor): 0<u<1\}
        \quad \text{and} \quad
        \{Z_{W_{x,r}}(\lfloor nu \rfloor): 0<u<1\},
        \end{equation}
        where $Z_{D_x}$ and $Z_{W_{x,r}}$ are the standardized versions of $D_x$ and $W_{x,r}$, respectively, with standardization as in ~\ref{subsec:definition}. Theorem 1 in \cite{kerseg} states that under regularity conditions on individual kernel matrices described in Appendix ~\ref{appendix:limiting-dist-details}, processes in ~\ref{eq:W-D-processes} converge to a Gaussian process in finite dimensional distributions for $r \neq 1$. 
    
        Since $Z_{D_{1}}(t), Z_{D_{2}}(t), Z_{W_{1},r}(t), Z_{W_{2},r}(t)$ are asymptotically gaussian, we may follow similar arguments in the proof for Proposition 3.4 in \cite{chen2015graph} for approximating the tail probabilities for maximum of Gaussian process to obtain the quantities in ~\ref{eq:fast-test}. Let $b$ be any given value:
        \begin{align}
        p_{Z^*_{D_x}}=\pr(Z^*_{D_x} > b)
        & = \pr(\max_{n_0\le t\le n_1} |Z_{D_{x}}(t)|> b),\label{eq:D-approx}\\
        p_{Z^*_{W_{x,r}}}=\pr(Z^*_{W_{x,r}} > b)&= \pr (\max_{n_0\le t\le n_1} Z_{W_{x,r}}(t)>b). \label{eq:W-approx}
        \end{align}
        The quantities in ~\ref{eq:D-approx} and ~\ref{eq:W-approx} can be obtained following similar procedure as in \cite{chen2015graph}. The details of the specific calculations are given in Appendix ~\ref{appendix:p-val-approx}.


\section{Experiments} \label{sec:performance}
    \textbf{Metrics.} We evaluate performance using \emph{Accurate Detection}. Under the alternative, where the true change point is $\tau$, a run is counted as accurate detection if a change point is detected and the estimated location $\hat{\tau}$ satisfies $|\hat{\tau}-\tau|\le 0.05n$. For methods that produce $p$-values, detection is declared when the $p$-value $\le 0.05$; for CPDstergm \citep{kei2025cpdseparable} and NBS \citep{wang2020optimalchangepointdetection}, detection is declared when the detected change-point is not null. We report the number of accurate detections over 100 runs. Under $H_0$, where no change point is present, this metric reduces to the number of false detections.

    In our experiments, we choose $K_1$ and $K_2$ to be Gaussian kernel with median heuristic \cite{JMLR:gretton-MMD} and graphlet kernel \cite{pmlr-v5-graphlet-kernel}. We choose these two kernels because they are diverse kernels that have different strengths. In practice, other kernels could also be applied. For permutation-based methods, we chose number of permutations B=1000. To assess the benefit of kernel aggregation, we compare the proposed KAP-CPD and KAPf-CPD methods with single-kernel baselines: GKCP (Gaussian) and GKCP (graph). We further compare with two existing methods for dynamic network change-point detection: NBS \cite{wang2020optimalchangepointdetection}, a CUSUM-based procedure assuming a Bernoulli network model, and CPDstergm \cite{kei2025cpdseparable}, a method based on separable temporal exponential random graph models (STERGM).
    
    \subsection{Power Comparison: Independence}
    We first consider settings in which edges are generated independently conditional on the model parameters. Let $N$ denote the number of nodes, $K$ the number of communities, $p$ the Erd\H{o}s--R\'enyi (ER) connection probability, and $\Lambda$ the SBM block connectivity matrix,$p_{\text{within}}$ and $p_{\text{across}}$ are the diagonal and off-diagonal entries of $\Lambda$. We consider two classes of models: ER/SBM and degree-corrected SBM (DCSBM).

\subsubsection{ER/SBM settings.}
We consider three independent-edge settings, all with sequence length $n=100$ and change point $\tau=50$.

    \textbf{ER:} $N=50$ and baseline $p=0.5$. After $\tau$, the connection probability increases by a signal level in $\{0,0.01,0.015,0.02,0.025,0.03,0.04,0.07\}.$
    
    \textbf{SBM:} $N=50$, $K=5$, $p_{\mathrm{within}}=0.5$, and $p_{\mathrm{across}}=0.3$. After $\tau$, $p_{\mathrm{within}}$ increases by
    $\{0,0.01,0.02,0.04,0.07,0.1,0.2\},$ corresponding to stronger within-community connectivity.
    
    \textbf{Sparse SBM:} $K=2$ balanced communities and $N$ ranges from $50$ to $200$. At $\tau$, the block matrix changes from
    $
    \Lambda_1=
    \begin{bmatrix}
    0.05 & 0.03\\
    0.03 & 0.05
    \end{bmatrix}
    \quad \text{to} \quad
    \Lambda_2=
    \begin{bmatrix}
    0.06 & 0.02\\
    0.02 & 0.06
    \end{bmatrix}.
    $
The results are shown in Figure~\ref{fig:simu:prob}.

\subsubsection{DCSBM settings.}
We next consider degree-corrected stochastic block models, where
$\Pr(A_{ij}=1)=\theta_i\theta_j\Lambda_{z_i z_j}.$
The degree parameters $\theta_i$ allow nodes in the same community to have different expected degrees, capturing hub-like behavior commonly observed in real networks. We consider three DCSBM scenarios, all with $n=100$, $\tau=50$, and $K=2$.

    \textbf{Degree-profile change:} $N=50$. After $\tau$, half of the nodes increase their degree parameters by the signal level, while the remaining half decrease by the same amount. The signal ranges over
    $\{0,0.1,0.12,0.15,0.17,0.18,0.2,0.25,0.3\}.$

    \textbf{Hub emergence:} $N=50$ and
    $
    \Lambda=
    \begin{bmatrix}
    0.05 & 0.03\\
    0.03 & 0.05
    \end{bmatrix}.
    $
    After $\tau$, a subset of nodes, indexed by the hub size, becomes hubs by increasing their degree parameter from $1$ to $1.3$. Hub size ranges from 0 to 16.
    
    \textbf{Block-probability change :} Node-specific degree parameters are generated as
    $
    \theta_i=y_i/(N^{-1}\sum_{j=1}^N y_j),
    \quad y_i\sim \mathrm{LogNormal}(0,0.25).
    $
    At $\tau$, the block matrix changes from
    $
    \Lambda_1=
    \begin{bmatrix}
    0.06 & 0.03\\
    0.03 & 0.06
    \end{bmatrix}
    \quad \text{to} \quad
    \Lambda_2=
    \begin{bmatrix}
    0.07 & 0.02\\
    0.02 & 0.07
    \end{bmatrix}.
    $ N ranges from 40 to 100.

    \begin{figure}[htbp]
        \centering
        \includegraphics[scale=0.52]{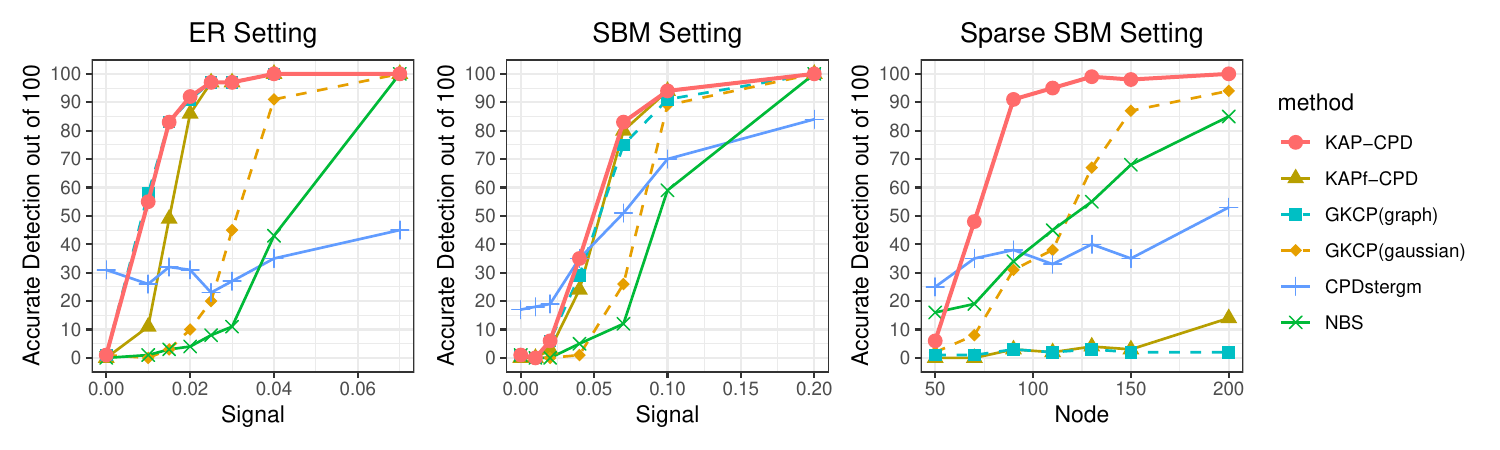}
        \caption{Accurate Detection in Probability Change Settings}
        \label{fig:simu:prob}
        \centering
        \vspace{10pt}
        \includegraphics[scale=0.52]{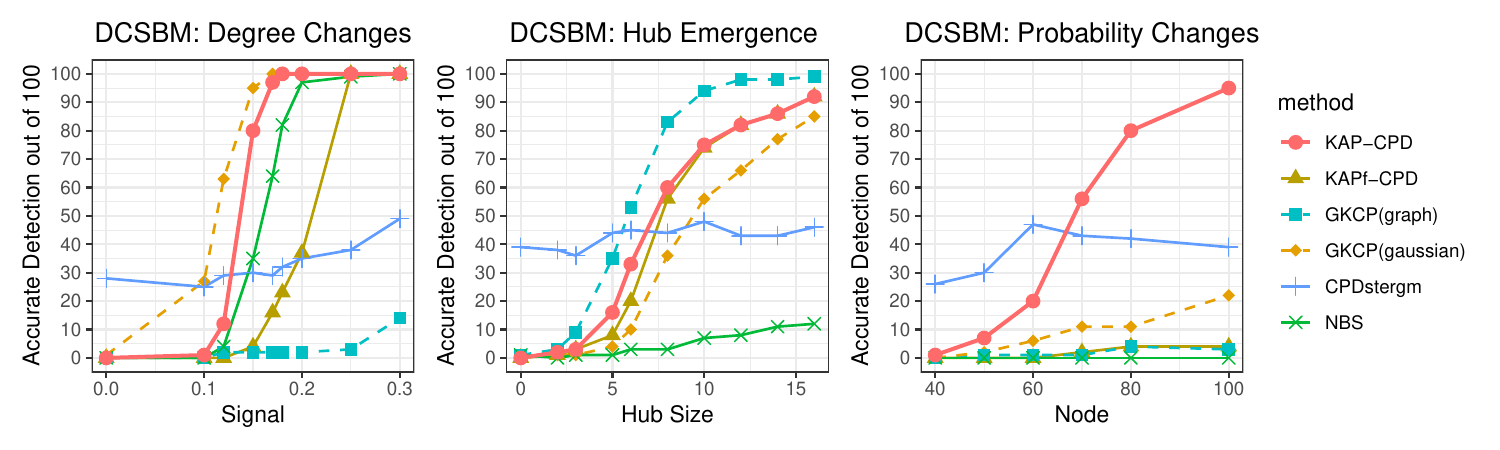}
        \caption{Accurate Detection in DCSBM Settings}
        \label{fig:simu:DCSBM}
    \end{figure}

    \subsection{Power Comparison: Dependence}
        We further examine settings with edge-level and temporal dependence. These forms of dependence are important in real-world networks, where transitivity or ``friends-of-friends'' effects are common: if two individuals are connected, their neighbors may also be more likely to form connections. Such settings violate the assumptions of NBS \cite{wang2020optimalchangepointdetection}, which relies on independence across both edges and observations. To assess robustness beyond independent Bernoulli network models, we consider Exponential random graph model (ERGM) and Random Geometric Graph (RGG) settings, which introduce edge-level dependence within each network, as well as STERGM settings, which introduce both edge-level dependence and temporal dependence.
        
            \textbf{ERGM Setting: Triangle Formation Change.} $n=100, N=50.$ Networks are generated from an ERGM with parameters $(\text{edge}=-2,\text{triangle}=0.1)$. After $\tau=50$, the triangle parameter increases by $[0.05,0.06,0.07,0.08,0.09,0.1]$. This models a change in higher-order dependence: nodes sharing a common neighbor become more likely to connect, leading to increased transitivity.
            
            \textbf{RGG Setting: Connection Radius Change.} $n=100, N=50.$ Networks are generated from an RGG with connection radius = $0.9\sqrt{\log(N)}/\pi N.$
            After $\tau=50$, the radius is multiplied by $\{1,1.02,1.04,1.06,1.08,1.11,1.12,1.14,1.2,1.26\}$. This models an expansion in the spatial range of interaction: nodes connect across larger distances, resulting in increased clustering.
            
            \textbf{STERGM Setting: Triangle Formation Change.} $n=100, N=50.$ Networks are generated from an STERGM with formation parameters $(\text{edge}=-1,\text{triangle}=-2)$ and persistence parameters $(\text{edge}=-1,\text{triangle}=-2)$. After $\tau=50$, the triangle parameter in persistence model increases by 
            `signal', and triangle parameter in formation increases by $3 \times \text{signal}$. signal ranges from 0.01 to 0.75. After the change-point, triangle patterns in networks become less discouraged. It also establishes temporal dependence across network observations: each network $G_t, t>1$ is dependent on $G_{t-1}$, which breaks Assumption ~\ref{assumption:exchangeability}.
        
        \begin{figure}[htbp]
            \centering
            \includegraphics[scale=0.52]{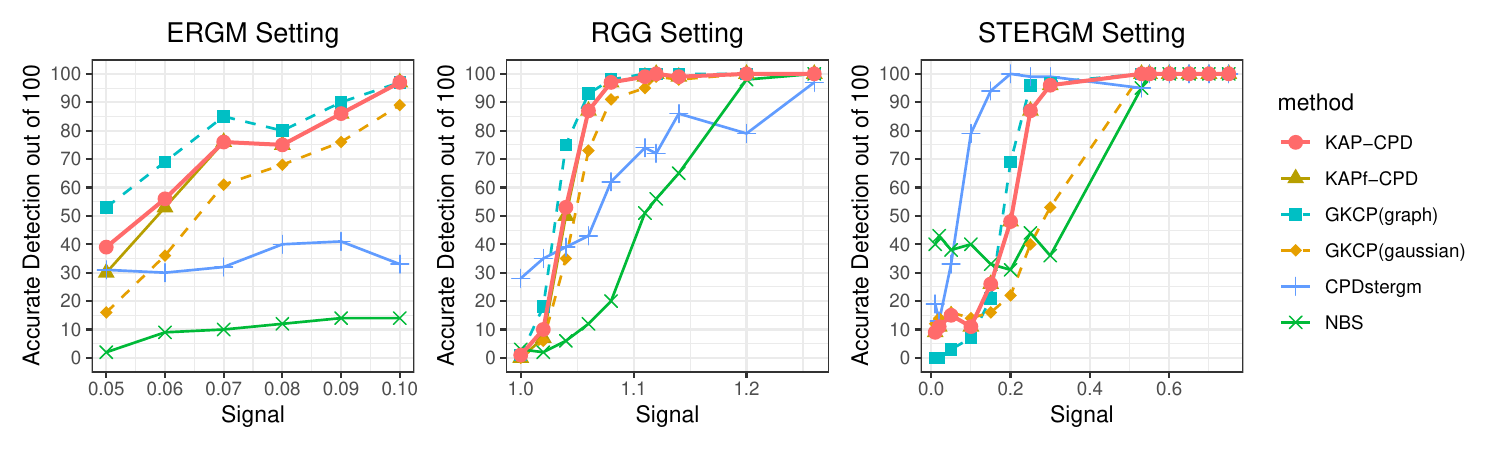}
            \caption{Accurate detection in settings with dyadic/observation dependence.}
            \label{fig:simu:dependence}
        \end{figure}
    We see that KAP-CPD combines kernels efficiently. The two dashed lines represent each of the single kernel baselines. We see that the respective single kernel baselines shows strong preference for certain settings: being top performers in some while having almost no power in others. KAP-CPD mitigates this issue by remaining close to the better performing kernel in their respective strengths. For settings such as ER, SBM, STERGM, KAP-CPD benefit from graphlet kernels and improves from single gaussian kernel baseline. For settings such as DCSBM degree changes and sparse SBM, KAP-CPD stays close to the better performing gaussian kernel baseline. It is worth noting that in sparse SBM and DCSBM probability changes setting, KAP-CPD shows much better improvement from both baselines. In the STERGM setting, CPDstergm is expected to perform well because it is based on the correctly specified parametric model; nevertheless, KAP-CPD achieves comparable performance in the medium signal ranges of $[0.25,0.75]$ without imposing this model assumption.
    
    \subsection{Empirical Size of Proposed Tests}\label{experiments: empirical-size}
    \begin{wraptable}{r}{0.56\textwidth}
        \vspace{-15pt}
        \footnotesize
        \caption{Type I error for various tests under $H_0$. n=100, significance level 0.05. Inflated type-I errors $>0.1$ are in red. }
        \label{tab:empirical_size}
        \centering
        \setlength{\tabcolsep}{3pt}
        \begin{tabular}{cccccc}
        \toprule
        \multicolumn{6}{c}{Empirical sizes over 1000 runs} 
        \\
        \cmidrule(r){1-6} 
         & KAP-CPD & KAPf-CPD & GKCP & NBS & CPDstergm\\
        \cmidrule(r){1-6} 
        ER  & 0.043 & 0.04 & 0.04 & \bad{0.233} & \bad{0.169} \\
        SBM  & 0.052 & 0.062 & 0.045 & 0.001 & 0.093 \\
        RGG & 0.051 & 0.04 & 0.051 & \bad{0.314} & \bad{0.291} \\
        ERGM & 0.055 & 0.035 & 0.052 & \bad{0.254} & \bad{0.237}\\
        \bottomrule
        \end{tabular}
     \end{wraptable}
    We evaluate empirical size at the nominal level $0.05$ with $n=100$. As shown in Table~\ref{tab:empirical_size}, KAP-CPD, KAPf-CPD, and GKCP (Gaussian) maintain stable finite-sample type-I error control, with empirical sizes close to the nominal level. All three kernel based methods uses fixed default values for all parameters. We see that they are well calibrated under Assumption~\ref{assumption:exchangeability}. In contrast, NBS and CPDstergm exhibit inflated type-I error in several settings. Since both methods rely on tuning parameters that are difficult to calibrate without ground truth, we use the default settings recommended in the original papers in our experiments.
    
    \subsection{Runtime Comparison}
        We next compare the computational efficiency of the competing methods. To obtain stable measurements, all methods were run on a MacBook Pro with an Apple M1 processor and 8GB RAM with no parallelization, although we note that there could be potential speedups with GPU or parallelization.
        \begin{wraptable}{r}{0.57\textwidth}
        \footnotesize
        \caption{Single Runtime comparison (in seconds) on different sequence sizes. Best performing method is in \textbf{bold}.}
        \centering
        \setlength{\tabcolsep}{2pt}
        \begin{tabular}{cccccc}
        \toprule
        \multicolumn{6}{c}{Apple M1} 
        \\
        \cmidrule(r){1-6} 
         & KAP-CPD & KAPf-CPD & GKCP & NBS & CPDstergm\\
        \cmidrule(r){1-6} 
        $n=100$  & 13.728  & 0.892 & 6.039  & \textbf{0.324} & 36.963  \\
        $n=500$  & 284.401 & \textbf{1.341} & 129.142 & 18.460 & 3286.846 \\
        $n=1000$ & 1087.797 & \textbf{9.605} & 520.512 & 94.325 & - \\
        $n=2000$ & 4554.776 & \textbf{36.229} & 2119.261 & 414.304 & -\\
        \bottomrule
        \end{tabular}
        \label{tab:runtime}
    \end{wraptable}
         For the permutation-based methods, the number of permutations B = 1000. As shown in Table~\ref{tab:runtime}, the proposed fast test can easily scale to long sequences like $n=2000$. Although NBS is computationally efficient for shorter sequences, its runtime increases substantially as $n$ grows. CPDstergm shows limited scalibility in this experiment and was not run for $n=1000$ and above.
        
\section{Real Data Example} \label{sec:real}
    \subsection{Enron Email Networks}\label{subsec:enron}
   The Enron Email Network dataset, accessed through `igraphdata' package \cite{enron-data} records the communication patterns among employees, primarily senior management, of the Enron Corporation between May 1999 and June 2002 \cite{peel2014cpd}. This dataset has been widely studied in network change-point detection, for instance \cite{peel2014cpd,frechet-cpd,kei2025-decoder}.Our goal is to determine whether significant events are reflected in the temporal evolution of the email exchange network. We pre-process the dataset as follows. First, we aggregate the communication network by week. For any pair of nodes, we form an undirected edge if at least one message was exchanged between them during a given week. We then apply binary segmentation to identify multiple change points. The detected change points, along with their nearby real-world events, are summarized in Table \ref{tab:enron-network}, where the event dates are cross-referenced with documented Enron events in \cite{peel2014cpd,enron_events}.Overall, the key events identified, most notably the California blackouts (January 17, 2001) and the stock plunge (November 28, 2001) and bankruptcy filing (December 2, 2001), align strongly with the major change points detected summarized in Table \ref{tab:enron-network}. 

    \begin{table*}[ht]
    \centering
    \caption{The Detected Changed Points and Corresponding \textit{p}-Values of KAP-CPD, GKCP (Gaussian), NBS and Fréchet for the Enron Email Network Dataset. (*** signifies that p-value is smaller than 0.001. Date indicates the week of "year'month/date")}
    \label{tab:enron-network}
    \resizebox{\textwidth}{!}{
    \begin{tabular}{ccccc}
    \toprule
     \textbf{KAP-CPD} & \textbf{GKCP } & \textbf{NBS}  & \textbf{Fréchet} & \textbf{Nearby Events} \\
    \hline
       99'6/28 -- 7/4 & 99'7/19--25 & - & 99'7/12--18&Enron CEO Exempted from  \\
     *** & 0.001 &  & 0.04&Code of Ethics (99'6/28) \\ 
    \hline
       99'12/6 -- 12 &99'11/29 -- 12/5 &  -  & - &\textbf{Launch of `EnronOnline' (99' 11/27)} \\
     ***  & 0.012 &  & \\ 
    \hline
       00'1/10 -- 16 & 99'12/20 -- 26 &99'12/20 -- 26 &  99'12/20 -- 26 &Launch of EBS (00' 1/17)\\
     ***  & *** & &*** & Annual meeting + Stock new high (00' 1/20) \\ 
    \hline
      00'4/10 -- 16 & 00'4/10 -- 16 &  -  & - & Conference Call with Stock Analysts\\
     *** & *** &  & \\ 
    \hline
      00'7/3 -- 9 & 00'7/3 -- 9 &  00'7/3 -- 9  & 00'6/12-18 & EBS-Blockbusters Partnership (00' 7/11)\\
     *** & *** &  & ***\\ 
    \hline
      -& - & - & 00'8/14 -- 20 & Stocks All Time High (00' 8/23)\\
      &  &  &*** \\ 
    \hline
       00'9/25 -- 10/1  & 00'9/25 -- 10/1 & 00'10/9 -- 15  & -&Enron Attornies discuss Belden’s strategies (00'10/3)\\
      0.0038 & *** &  &  \\ 
    \hline
      00'11/6 -- 12 & 00'11/6 -- 12 & -  & -&FERC Exonerates Enron for any \\
     *** & *** &   && wrongdoing in California (00'11/1) \\ 
    \hline
      \textbf{01'1/15 -- 21} & \textbf{01'1/15 -- 21} &  \textbf{01'1/29 -- 2/5}  & -&\textbf{California Major Blackouts (01'1/17)}\\
     *** & *** &  &  \\ 
    \hline
       01'4/9 -- 15 & 01'4/9 -- 15 &  -  & - &\textbf{The Infamous Conference Call (01'4/17)}\\
     0.006 & *** &  && \textbf{with investors} \\ 
    \hline
       01'6/4 -- 10 & 01'6/4 -- 10 & 01'7/2 -- 8   & - &Quarterly Conference Call+Energy Crisis Ends \\
         *** & *** &  &  \\ 
    \hline
       01'7/23 -- 29 & - & -  & - &CEO meets with investors in NY (01' 7/24)\\
     0.005 &  &  &  \\ 
    \hline
       01'9/17 -- 23 & 01'9/3 -- 9 & -  & - & CEO sells \$15.5 million of stock (01'9)\\
      *** & *** &  & \\ 
     \hline
      \textbf{01'11/26 -- 12/2} & \textbf{01'11/26 -- 12/2} & \textbf{01'12/3 -- 9}  & \textbf{01'12/17 -- 23} &\textbf{Stock plunges below \$1 (01'11/28)}\\
     *** & *** &  &***&\textbf{File Bankruptcy (01' 12/2)} \\ 
     \hline
      02'2/4 -- 10 & 02'2/4 -- 10 &  -  & -&\textbf{Cooper takes over as CEO (02'1/30)}\\
      *** & *** &  &&\textbf{Skilling testifies before Congress (02'2/7)}   \\ 
      \hline
      02'3/18 -- 24 & 02'3/25 -- 31 &  -  & -&\textbf{Arthur Andersen (Enron Auditor) Indicted (02'3/14)}\\
     *** & *** &  &  \\ 
    \bottomrule
    \end{tabular}}
    \end{table*}

    \subsection{Seizure Detection in Functional Connectivity Networks}\label{subsec:seizure}
    We apply our method in the database “Detect seizures in intracranial electro-encephalogram (iEEG) recordings” provided originally by the UPenn and Mayo Clinic \cite{seizure-detection}  (\url{https://www.kaggle.com/c/seizure-detection}), which consists of iEEG recordings of 12 subjects (eight human subjects and four canines).
    \begin{wraptable}{r}{0.5\textwidth}
    \centering
    \footnotesize
    \setlength{\tabcolsep}{3pt}
    \renewcommand{\arraystretch}{0.95}
    \caption{Average $|\tau -\hat \tau| (\text{SD})$ for seizure detection data over 30 random seeds. Smaller $|\tau -\hat \tau| (\text{SD})$ are in \textbf{bold}.}
    \label{tab:seizures}
    \begin{tabular}{lllccc}
    \toprule
    Subject & $n$ & $\tau$ & KAPf-CPD &   NBS \\
    \midrule
    Dog 1  & 596  & 178 & \textbf{0.6 (1.28)} &  \textbf{0.6 (1.59)}  \\
    Dog 2  & 1320 & 172 & \textbf{1.07 (2.03)} &  2.27 (3.67)  \\
    Dog 3  & 5240 & 480 & \textbf{0.67 (0.84)} &  0.8 (1.63)  \\
    Dog 4  & 3047 & 257 & 1.77 (2.47) &  \textbf{1.33 (0.92)}  \\
    \midrule
    Patient 1 & 174  & 70  & \textbf{1.87 (2.26)}  & 2.47 (3.05) \\
    Patient 2 & 3141 & 151 & \textbf{0.2 (0.41)}  & 1.47 (1.36) \\
    Patient 3 & 1041 & 327 & \textbf{1.67 (2.04)}  & 2.8 (3.42) \\
    Patient 4 & 210  & 20  & \textbf{0.97 (2.03)}    & 28.8 (48.5)  \\
    Patient 5 & 2745 & 135 & \textbf{0.87 (1.25)}  & 2.4 (3.67) \\
    Patient 6 & 2997 & 225 & \textbf{0.43 (0.68)}  & 1.33 (0.92)  \\
    Patient 7 & 3521 & 282 & \textbf{0.9 (1.03)}  & 1.33 (1.99)  \\
    Patient 8 & 1890 & 180 & \textbf{1.63 (2.67)}  & 3.47 (4.42)\\
    \bottomrule
    \end{tabular}
    \vspace{-8pt}
\end{wraptable}
     We obtained the preprocessed dataset in \cite{cpd_on_graphs-real-data} and \cite{rank-based-cp} who represented the iEEG data as functional connectivity networks calculated from Pearson correlation in the high-gamma band (70-100Hz). Functional connectivity networks are weighted graphs, where electrodes are represented by vertices, and edge weights correspond to the coupling strength of the nodes. Because of the length of the sequence is long, we applied the two relatively faster methods, KAPf-CPD combining Gaussian and Graphlet kernel, and NBS \cite{wang2020optimalchangepointdetection}. Following the experiment setup in \cite{cpd_on_graphs-real-data,rank-based-cp}, to obtain a stationary sequence of graphs from seizure period and normal activity period, we select all graphs of each class and randomly shuffle them within their class, and then concatenate them into $\{G_t\}_{t=1,\dots,n}$, creating an artificial change-point localization benchmark with known change-point location at $\tau$ for each subject. We report the average and standard deviation of the differences between true change-point and estimated change-point $|\hat \tau - \tau|$ over 30 random seeds for shuffling. 
    The performance comparison is given in Table \ref{tab:seizures}. We note that because the change-point location is highly unbalanced, we set the cutoff to be $n_0=0.04n$ and $n_1=0.95n$. KAPf-CPD achieves smaller average localization error for most subjects.


\section{Conclusion} \label{sec:dicussion}
     
In this work, we propose a novel test procedure KAP-CPD which aggregates information from different kernels via a Mahalanobis-type combination. This construction allows the procedure to capture similarity patterns from diverse perspectives, leading to broader adaptivity to different alternatives. To further enhance computational efficiency, we also introduce KAPf-CPD, which provides analytical $p$-values approximation in place of permutation-based $p$-values. In practice, KAPf-CPD is well suited for rapid initial screening of candidate change-points, while KAP-CPD can subsequently be applied for a confirmation of detected changes or in case of ambiguous results. Through extensive simulation studies, we demonstrate that the proposed approach outperforms single-kernel baselines and several state-of-the-art network change-point detection methods in detection power, runtime speed and type-I error control. Together, these findings suggest that kernel aggregation provides a flexible and powerful framework for change-point detection in dynamic networks.

    \bibliographystyle{unsrtnat}
    \bibliography{references}

    \newpage
    \appendix 
    \section{Illustrative example for comparing with other forms of kernel aggregation}
    Most existing kernel-combination methods have been developed for two-sample or independence testing \cite{2023Boosting-kernel-two-sample-test,2023mmdfuse,zhou2025dual-kernel-combination,JMLR:MMD-Agg}. To isolate the effect of the aggregation mechanism, we consider a fixed-split version of our procedure, which reduces the change-point problem to a two-sample testing problem. Specifically, suppose the candidate change point $\tau$ is known. We test
        \[
        H_0: X_1,\ldots,X_n \sim F_0
        \]
        against
        \[
        H_1: X_1,\ldots,X_\tau \sim F_0,
        \qquad
        X_{\tau+1},\ldots,X_n \sim F_1.
        \]
    In this setting, we evaluate the statistic $S(t)$ only at $t=\tau$ and obtain the $p$-value by permutation.
    
    We compare this fixed-split version of KAP-CPD with MMD-FUSE \cite{2023mmdfuse}, a two-sample testing method that aggregates normalized MMD statistics across multiple kernels through a weighted soft maximum. 
    
    We consider sparse mean shifts under heavy-tailed noise. For each setting, we generate
    \[
    X_1,\ldots,X_\tau \sim t_3(0,I_d),
    \qquad
    X_{\tau+1},\ldots,X_n \sim t_3(\mu,I_d),
    \]
    with $\tau=50$ and $n=100$. The mean-shift vector is sparse: $\mu_j=\delta$ for $j\le 10$ and $\mu_j=0$ otherwise. We vary $(d,\delta)$ pair. For each setting, we run 100 replications at level $\alpha=0.05$. For MMD-FUSE, we compare Gaussian (G), Laplacian (L), and Gaussian+Laplacian (G+L) kernel collections, each using 10 bandwidths. For KAP-CPD, we combine Gaussian and Laplacian kernels with the median heuristic. Both methods use $B=2000$ permutations.
    
    Table~\ref{tab:mmd_fuse_comparison} reports rejection counts out of 100 replications. This setting favors the Laplacian kernel. In all three settings, the fixed-split KAP-CPD aggregation benefits from adding the Laplacian kernel and remains close to the stronger single-kernel performance. By contrast, MMD-FUSE(G+L) shows limited improvement over MMD-FUSE(G) in the more difficult settings. This experiment is not intended as a comprehensive comparison with two-sample testing methods; rather, it illustrates that the proposed aggregation strategy can preserve useful information from a favorable kernel in a fixed-split setting.

    \begin{table}[ht]
    \centering
    \caption{Rejection counts out of 100 runs for sparse mean-shift alternatives under different kernel-combination strategies.}
    \label{tab:mmd_fuse_comparison}
    \small
    \setlength{\tabcolsep}{5pt}
    \renewcommand{\arraystretch}{0.95}
    \begin{tabular}{cccccc}
    \toprule
    $d$ & $\delta$ & MMD-FUSE (G) & MMD-FUSE (L) & MMD-FUSE (G+L) & KAP-CPD (G+L) \\
    \midrule
    30 & 0.6 & 85 & 97 & 93 & 97 \\
    30 & 0.5 & 60 & 89 & 74 & 80 \\
    40 & 0.5 & 51 & 83 & 57 & 80 \\
    50 & 0.4 & 18 & 52 & 24 & 39 \\
    50 & 0.5 & 42 & 80 & 55 & 79 \\
    60 & 0.5 & 48 & 76 & 60 & 64 \\
    \bottomrule
    \end{tabular}
    \end{table}

    \section{Lemma1}
        \begin{lemma}\label{lemma:first-two-moments}
        Under the permutation null, for any $a,b\in\{1,2\}$,
        \begin{align*}
        \E[\alpha_a(t)] &= \E[\beta_a(t)]= \bar k_a,\\
        \cov(\alpha_a(t),\alpha_b(t))
        &=
        \frac{2A_{ab}p_1(t)+4B_{ab}p_2(t)+C_{ab}p_3(t)}
             {t^2(t-1)^2}
        -\bar k_a\bar k_b,\\
        \cov(\beta_a(t),\beta_b(t))
        &=
        \frac{2A_{ab}q_1(t)+4B_{ab}q_2(t)+C_{ab}q_3(t)}
             {(n-t)^2(n-t-1)^2}
        -\bar k_a\bar k_b,\\
        \cov(\alpha_a(t),\beta_b(t))
        &=
        \frac{C_{ab}}{n(n-1)(n-2)(n-3)}
        -\bar k_a\bar k_b.
        \end{align*}
        Here
        \begin{align*}
        \bar k_a &= \frac{1}{n(n-1)}\sum_{i=1}^n\sum_{j\ne i} k_{aij}, \quad A_{ab} = \sum_{i=1}^n\sum_{j\ne i} k_{aij}k_{bij},\\
        B_{ab} &= \sum_{i=1}^n\sum_{j\ne i}\sum_{\substack{u=1\\u\ne i,j}}k_{aij}k_{biu}, \quad C_{ab} = \sum_{i=1}^n\sum_{j\ne i}\sum_{\substack{u=1\\u\ne i,j}}
        \sum_{\substack{v=1\\v\ne i,j,u}} k_{aij}k_{buv},
        \end{align*}
        and
        \begin{align*}
        p_1(t) &= \frac{t(t-1)}{n(n-1)},
        & p_2(t) &= p_1(t)\frac{t-2}{n-2},
        & p_3(t) &= p_2(t)\frac{t-3}{n-3},\\
        q_1(t) &= \frac{(n-t)(n-t-1)}{n(n-1)},
        & q_2(t) &= q_1(t)\frac{n-t-2}{n-2},
        & q_3(t) &= q_2(t)\frac{n-t-3}{n-3}.
        \end{align*}
        \end{lemma}

    \begin{proof}
    Take kernels $K_1,K_2 \in \mathbbm{R}^{n \times n}$, 
    \begin{align*}
        \alpha_1(t)&=\frac{1}{t(t-1)}{\sum_{i=1}^n\sum_{j\neq i}} k_{1ij} \mathbbm{1}(i\le t,j \le t)\\
        \beta_1(t)&= \frac{1}{(n-t)(n-t-1)}{\sum_{i=1}^n\sum_{j\neq i}} k_{1ij} \mathbbm{1}(i> t,j > t)\\
        \text{Under } &\text{the permutation null, without loss of generallity, take } a=1,b=2: \\
        \E[\alpha_1(t)]&= \E\left[\frac{1}{t(t-1)}{\sum_{u=1}^n\sum_{v\neq u}} k_{1ij} \mathbbm{1}(i\le t,j \le t)\right]\\
        &= \frac{1}{t(t-1)}{\sum_{i=1}^n\sum_{j\neq i}} k_{1ij} \E[\mathbbm{1}(i\le t,j \le t)]\\
        &=\frac{1}{t(t-1)}{\sum_{i=1}^n\sum_{j\neq i}} k_{1ij} \frac{t(t-1)}{n(n-1)}=\bar k_1\\
        \E[\beta_2(t)]&= \E\left[\frac{1}{(n-t)(n-t-1)}{\sum_{u=1}^n\sum_{v\neq u}} k_{2ij} \mathbbm{1}(i\le t,j \le t)\right]\\
        &= \frac{1}{(n-t)(n-t-1)}{\sum_{i=1}^n\sum_{j\neq i}} k_{2ij} \E[\mathbbm{1}(i\le t,j \le t)]\\
        &=\frac{1}{(n-t)(n-t-1)}{\sum_{i=1}^n\sum_{j\neq i}} k_{2ij} \frac{(n-t)(n-t-1)}{n(n-1)}=\bar k_2\\
        \E[\alpha_1(t)\beta_2(t)] &= \E\left[
        \frac{1}{t(t-1)}
        \sum_{i=1}^n \sum_{j\neq i} k_{1ij}\mathbbm{1}(i\le t,j\le t)
        \right. \notag\\
        &\qquad\left.
        \times
        \frac{1}{(n-t)(n-t-1)}
        \sum_{u=1}^n \sum_{v\neq u} k_{2uv}\mathbbm{1}(u>t,v>t)
        \right]\\
        &= \frac{\sum_i\sum_{j\neq i}\sum_{u\neq i\neq j}\sum_{v\neq u\neq i\neq j} k_{1ij}k_{2uv}P(i\le t, j \le t, u> t, v>t)}{t(t-1)(n-t)(n-t-1)}\\
        &= \frac{\sum_i\sum_{j\neq i}\sum_{u\neq i\neq j}\sum_{v\neq u\neq i\neq j} k_{1ij}k_{2uv}\frac{t(t-1)(n-t)(n-t-1)}{n(n-1)(n-2)(n-3)}}{t(t-1)(n-t)(n-t-1)} \\
        &=\frac{ \sum_{i=1}^n\sum_{j\neq i}\sum_{u\neq i\neq j}\sum_{v\neq u\neq i\neq j} k_{1ij}k_{2uv}}{n(n-1)(n-2)(n-3)}\\
        &= \frac{C_{12}}{n(n-1)(n-2)(n-3)}\\
        \E[\alpha_1(t)\alpha_2(t)]&= \E\left[\frac{1}{t(t-1)}{\sum_{u=1}^n\sum_{v\neq u}} k_{1ij} \mathbbm{1}(i\le t,j \le t) \times \frac{1}{t(t-1)}{\sum_{u=1}^n\sum_{v\neq u}} k_{2uv} \mathbbm{1}(u\le t,v \le t)\right]\\
        &= \frac{1}{t^2(t-1)^2}\sum_i\sum_{j\neq i}\sum_{u}\sum_{v\neq u} k_{1ij}k_{2uv}P(i\le t, j \le t, u \le t, v \le t)\\
        &= \frac{2}{t^2(t-1)^2}\sum_i\sum_{j\neq i}k_{1ij}k_{2ij}\frac{t(t-1)}{n(n-1)}\\
        &+ \frac{4}{t^2(t-1)^2}\sum_i\sum_{j\neq i}\sum_{u\neq i\neq j} k_{1ij}k_{2iu}\frac{t(t-1)(t-2)}{n(n-1)(n-2)}\\
        &+\frac{1}{t^2(t-1)^2}\sum_i\sum_{j\neq i}\sum_{u\neq i\neq j}\sum_{v \neq i \neq j \neq u} k_{1ij}k_{2uv}\frac{t(t-1)(t-2)(t-3)}{n(n-1)(n-2)(n-3)}\\
        &=\frac{2A_{12}p_1(t)+4B_{12}p_2(t)+C_{12}p_3(t)}
             {t^2(t-1)^2}
     \end{align*}
     $\E[\beta_1(t)\beta_2(t)]$ can be obtained similarly. We can obtain corresponding quantities by taking $\cov(\alpha_1(t),\alpha_2(t))=\E[\alpha_1(t)\alpha_2(t)]-\E[\alpha_1(t)]\E[\alpha_2(t)]$.
    \end{proof}

    \section{Proof of Theorem \ref{thm:decomposition}}\label{appendix:proof_thm1}
    \begin{proof}
        Through tedious algebra by plugging in the values in Lemma ~\ref{lemma:first-two-moments}, we can show that:
        \begin{align*}
        \Cov(W_1(t)-W_2(t),D_1(t)-D_2(t)) &= \Cov(W_1(t)-W_2(t),D_1(t)+D_2(t))\\
        &=\Cov(W_1(t)+W_2(t),D_1(t)-D_2(t))\\
        &=\Cov(W_1(t)+W_2(t),D_1(t)+D_2(t))=0.
        \end{align*}
        By verifying that:
        $$\cov(W_1(t),D_1(t))=\cov(W_1(t),D_2(t))=\cov(W_2(t),D_1(t))=\cov(W_2(t),D_2(t))=0.$$
        
        Next we want to show $\exists c_1,c_0 \quad s.t \quad \cov(c_1W_1(t) + c_0W_2(t), W_1(t)-W_2(t))=0.$
         \begin{align*}
         \quad \Cov(c_1W_1(t) + c_0W_2(t), W1(t)-W_2(t)) &= c_1\Var(W_1(t)) -c_1\Cov(W_1(t),W_2(t)) \\
         &+ c_0 \Cov(W_1(t),W_2(t)) - c_0\Var(W_2(t))=0\\
         \text{So we want to find } c_1,c_0 \quad  s.t \quad c_1 \Var(W_1(t)) &- c_0\Var(W_2(t)) = (c_1-c_0) \Cov(W_1(t), W_2(t))
         \end{align*}
         Then take $c_0=1$, we have:
         $$c_1(t)=\frac{\Var(W_2(t))-\Cov(W_1(t),W_2(t))}{\Var(W_1(t))-\Cov(W_1(t),W_2(t))}$$

         Similarly we can show that with $c_2(t)=\frac{\Var(D_2(t))-\Cov(D_1(t),D_2(t))}{\Var(D_1(t))-\Cov(D_1(t),D_2(t))}$, $\cov(c_2D_1(t) + D_2(t), D_1(t)-D_2(t))=0.$
         
         \begin{align*}
            \vec U
            &=
            \begin{bmatrix}
            W_1(t)-W_2(t) \\
            D_1(t)-D_2(t) \\
            c_1(t) W_1(t)+W_2(t) \\
            c_2(t) D_1(t)+D_2(t)
            \end{bmatrix}
            = A_t B_t \vec a_t
            = V_t \vec a_t,
        \end{align*}
        Where
        \begin{align*}
            A_t &=
            \begin{bmatrix}
            1 & -1 & 0 & 0 \\
            0 & 0  & 1 & -1\\
            c_1(t) & 1 & 0 & 0 \\
            0 & 0 & c_2(t) & 1
            \end{bmatrix},
            \vec a_t =
            \begin{bmatrix}
            \alpha_1(t) \\
            \beta_1(t) \\
            \alpha_2(t) \\
            \beta_2(t)
            \end{bmatrix},\\
             B_t &=
            \begin{bmatrix}
            \frac{t}{n} & \frac{n-t}{n} & 0 & 0 \\
            0 & 0  & \frac{t}{n} & \frac{n-t}{n}\\
            \frac{t(t-1)}{n(n-1)} & -\frac{(n-t)(n-t-1)}{n(n-1)} & 0 & 0 \\
            0 & 0 & \frac{t(t-1)}{n(n-1)} & -\frac{(n-t)(n-t-1)}{n(n-1)}
            \end{bmatrix}. \\
        \end{align*}
        
        We see that $S(t)$ = $(\vec a_t-\E[\vec a_t])^T \Sigma_t^{-1}(\vec a_t-\E[\vec a_t])= (V(\vec a_t-\E[\vec a_t]))^T (V\Sigma_tV^T)^{-1}V(\vec a_t-\E[\vec a_t])$.
        And since:
        \begin{align*}
        &\Cov(W_1(t)-W_2(t),D_1(t)-D_2(t))= \Cov(W_1(t)-W_2(t),D_1(t)+D_2(t))\\
        &=\Cov(W_1(t)+W_2(t),D_1(t)-D_2(t))=\Cov(W_1(t)+W_2(t),D_1(t)+D_2(t))\\
        &=\Cov(c_1 W_1(t)+W_2(t),W_1(t)-W_2(t))=\Cov(c_2 D_1(t)+D_2(t),D_1(t)-D_2(t))\\
        &=0
        \end{align*}
        $(V\Sigma_tV^T)^{-1}$ is a $4\times 4$ diagonal matrix. Thus, $$S(t)= Z_{W_{\text{diff}}}^2(t)+Z_{D_{\text{diff}}}^2(t)+Z_{W_{\text{sum}}}^2(t)+Z_{D_{\text{sum}}}^2(t).$$
    \end{proof}
    
    \section{Proof of Theorem \ref{thm2}}
        Consider the vector:
        \begin{align*}
            \vec u &= \begin{bmatrix}
            W_1(t) - W_2(t) \\
            D_1(t) - D_2 (t)\\
            W_1(t) + W_2 (t)\\
            D_1 (t)+ D_2(t)
            \end{bmatrix} = \begin{bmatrix}
            1 & -1 & 0 & 0 \\
            0 & 0  & 1 & -1\\
            1 & 1 & 0 & 0 \\
            0 & 0 & 1 & 1
            \end{bmatrix} \\
            &\times \begin{bmatrix}
            \frac{t}{n} & \frac{n-t}{n} & 0 & 0 \\
            0 & 0  & \frac{t}{n} & \frac{n-t}{n}\\
            \frac{t(t-1)}{n(n-1)} & -\frac{(n-t)(n-t-1)}{n(n-1)} & 0 & 0 \\
            0 & 0 & \frac{t(t-1)}{n(n-1)} & -\frac{(n-t)(n-t-1)}{n(n-1)}
            \end{bmatrix} \times \begin{bmatrix}
            \alpha_1 (t)\\
            \beta_1(t)\\
            \alpha_2(t) \\
            \beta_2(t)
            \end{bmatrix}.
        \end{align*}

        Let
        \[
        \begin{aligned}
        a_W &= \Var(W_1(t)-W_2(t)), 
        & a_D &= \Var(D_1(t)-D_2(t)),\\
        b_W &= \Var(W_1(t))-\Var(W_2(t)), 
        & b_D &= \Var(D_1(t))-\Var(D_2(t)),\\
        c_W &= \Var(W_1(t)+W_2(t)), 
        & c_D &= \Var(D_1(t)+D_2(t)).
        \end{aligned}
        \]
        Then
        \[
        \Sigma_{\vec u}
        =
        \begin{bmatrix}
        a_W & 0 & b_W & 0 \\
        0 & a_D & 0 & b_D \\
        b_W & 0 & c_W & 0 \\
        0 & b_D & 0 & c_D
        \end{bmatrix}
        =
        \begin{bmatrix}
        A & B\\
        B & C
        \end{bmatrix},
        \]
        where
        \[
        A=\operatorname{diag}(a_W,a_D),\qquad
        B=\operatorname{diag}(b_W,b_D),\qquad
        C=\operatorname{diag}(c_W,c_D).
        \]
        
        $$\Sigma_{\vec u}= \begin{bmatrix}
        I & BC^{-1} \\
        0 & I \\
        \end{bmatrix} \times \begin{bmatrix}
        A-BC^{-1}B^T & 0 \\
        0 & C \\
        \end{bmatrix} \times \begin{bmatrix}
        I & 0 \\
        C^{-1}B^T & I \\
        \end{bmatrix}$$
        
        $$\det(\Sigma_{\vec u})=\det(C)\det(A-BC^{-1}B^T)$$
        
        In order for $\Sigma_{\vec u}$ to be invertible, we need C is invertible and $A-BC^{-1}B^T$ is invertible.
        
        $$C= \begin{bmatrix}
         c_W & 0  \\
         0 &c_D
        \end{bmatrix}$$
        
        $$A-BC^{-1}B^T= \begin{bmatrix}
        a_D - \frac{b_D^2}{c_D} & 0  \\
         0 & a_W - \frac{b_W^2}{c_W}
        \end{bmatrix}$$
        
        In other words we need:
        \begin{align}
         &1. \quad \Var(D_1(t)-D_2(t)) - \frac{(\Var(D_1(t))-\Var(D_2(t)))^2}{\Var(D_1(t)+D_2(t))} \neq 0 \label{cond:invert}\\
         &2. \quad \Var(W_1(t)-W_2(t)) - \frac{(\Var(W_1(t))-\Var(W_2(t)))^2}{\Var(W_1(t)+W_2(t))} \neq 0 \label{cond:invert2}\\
         &3. \quad \Var(D_1(t) + D_2(t)) > 0\\
         &4. \quad \Var(W_1(t) + W_2(t)) > 0 
        \end{align}
        
        \begin{align*}
         &\Var(D_1(t)-D_2(t)) - \frac{(\Var(D_1(t))-\Var(D_2(t)))^2}{\Var(D_1(t)+D_2(t))} \\
        & =  \frac{\Var(D_1(t)-D_2(t))\Var(D_1(t)+D_2(t))-(\Var(D_1(t))-\Var(D_2(t)))^2}{\Var(D_1(t)+D_2(t))}\\
         \end{align*}
        Equivalently we need: 
        \begin{equation}
            \Var(D_1(t)-D_2(t))\Var(D_1(t)+D_2(t))-(\Var(D_1(t))-\Var(D_2(t)))^2>0
        \end{equation}
         \begin{align*}
         \Var(D_1(t)-D_2(t))\Var(D_1(t)+D_2(t))
         &= \Var(D_1(t))^2 + 2(\Var(D_1(t))\Var(D_2(t))) \\
         &+ \Var(D_2(t))^2 - 4\Cov(D_1(t), D_2(t))^2\\
          (\Var(D_1(t))-\Var(D_2(t)))^2 &= \Var(D_1(t))^2 - 2\Var(D_1(t))\Var(D_2(t)) + \Var(D_2(t))^2
         \end{align*}
        
         So in order to satisfy \eqref{cond:invert} we need:
        $$\Var(D_1(t))\Var(D_2(t))-\Cov(D_1(t),D_2(t))^2 \neq 0 \rightarrow D_1(t), D_2(t) \text{ not perfectly linearly correlated.}$$
        
         Similarly for \eqref{cond:invert2}, we need: $$ \Var(W_1(t))\Var(W_2(t))-\Cov(W_1(t),W_2(t))^2 \neq 0 , \text{ or equivalently } W_1(t), W_2(t)$$  not perfectly linearly correlated.
         So the conditions are:
         \begin{align*}
         &1. \quad \Var(D_1(t))\Var(D_2(t))-\Cov(D_1(t),D_2(t))^2 \neq 0\\
         &2. \quad \Var(W_1(t))\Var(W_2(t))-\Cov(W_1(t),W_2(t))^2 \neq 0\\
         &3. \quad \Var(D_1(t) + D_2(t)) \neq 0\\
         &4. \quad \Var(W_1(t) + W_2(t)) \neq 0 \\
        \end{align*}

\section{Asymptotic p-value approximation and limiting distributions details}
    \subsection{limiting distribution details} \label{appendix:limiting-dist-details}
        Let $\tilde k_{ij}=k_{ij}-\bar k$ and $k_{i.}=\sum_{j=1,\ j\neq i}^n k_{ij}$. 
        According to Theorem~1 in \cite{kerseg}, if $K_1$ and $K_2$ each satisfy: 1. $\sum_{i=1}^n |\tilde k_{i.}|^s
            =
            o\left[
            \left\{\sum_{i=1}^n \tilde k_{i.}^2\right\}^{s/2}
            \right] \text{for all integers } s>2,
            $ and
            2. $
            \sum_{i,j=1}^n \tilde k_{ij}^2
            =
            o\left(
            \sum_{i=1}^n \tilde k_{i.}^2
            \right)
            $,
    we will have that each of $\{Z_{D_1}[nu]: 0<u<1\},\{Z_{D_2}[nu]: 0<u<1\}$
    converges to a Gaussian process in finite dimensional distributions, denoted as $\{Z^*_{D_x}[nu]: 0<u<1\}$. 
    We also have that each of $\{Z_{W_{x,r}}[nu]: 0<u<1\}$ converges to a Gaussian process in finite dimensional distributions when $r \neq 1$. In the case of Gaussian kernel and graphlet kernel with each $k_{ij} \in [0,1]$, since each $k_{ij}$ is of constant order, unless under unusual circumstances with significant outliers such as one observation $G_t$ dominates the row sums $\tilde k_{i.}$ through $\tilde k_{t.}$, we will have $|\tilde k_{ij}| =  |k_{ij}-\bar k|= O(1)$ and $|\tilde k_{i.}| = O(n)$ $\forall i$, which satisfy both conditions. 
   
    \subsection{P-value approximation details} \label{appendix:p-val-approx}
    According to \cite{chen2015graph}, when $n_{0}, n_{1}, n, b \rightarrow \infty$ in a way such that for some $0 < x_{0} < x_{1} < 1$ and $b_{0} > 0$, $n_{0}/n \rightarrow x_{0}$, $n_{1}/n \rightarrow x_{1}$, and $b/\sqrt{n} \rightarrow b_{0}$, as $n\rightarrow \infty$, we have:     
    \begin{align}
    	&\pr\left(\max_{n_{0}\leq t \leq n_{1}}|Z_{D_{x}}^{*}(t/n)| > b\right)  \sim 2b\phi(b)\int_{x_{0}}^{x_{1}}h_{D_{x}}^{*}(x)\nu\Big(b_{0}\sqrt{2h_{D_{x}}^{*}(x)}\Big)dx, \label{pvalapprox1}\\ 
        &\pr\left(\max_{n_{0}\leq t \leq n_{1}}Z_{W_{x,r}}^{*}(t/n) > b\right)  \sim b\phi(b)\int_{x_{0}}^{x_{1}}h_{W_{x,r}}^{*}(x)\nu\Big(b_{0}\sqrt{2h_{W_{x,r}}^{*}(x)}\Big)dx, \label{pvalapproxw}
    \end{align}
    where the function $\nu(\cdot)$ can be numerically estimated as $\nu(s) \approx \frac{(2/s)\left(\Phi(s/2)-0.5\right)}{(s/2)\Phi(s/2)+\phi(s/2)}$ according to \cite{siegmund2007statistics} with $\Phi(\cdot)$ and $\phi(\cdot)$ being the standard normal cumulative density function and probability density function, respectively. 
    \begin{align*}
    	h_{D_x}^{*}(x) &= \lim_{s\nearrow x}\frac{\partial\rho_{D_x}^{*}(s,x)}{\partial s}  = -\lim_{s\searrow x}\frac{\partial\rho_{D_x}^{*}(s,x)}{\partial s}.
    \end{align*}
    
    \begin{remark}
    	In practice, when using ~\eqref{pvalapprox1}-~\eqref{pvalapproxw} for finite sample, we use
    	\begin{align*}
    		&\pr\left(\max_{n_{0}\leq t \leq n_{1}}|Z_{D_x}(t)| > b\right)  \sim 2b\phi(b)\sum_{n_{0}\le t \le n_{1}}C_{D_x}(t)\nu\Big(b\sqrt{2C_{D_x}(t)}\Big),\\
            &\pr\left(\max_{n_{0}\leq t \leq n_{1}}Z_{W_{r,x}}(t) > b\right)  \sim 2b\phi(b)\sum_{n_{0}\le t \le n_{1}}C_{W_{r,x}}(t)\nu\Big(b\sqrt{2C_{W_{r,x}}(t)}\Big),\\
    	\end{align*}
    	where
    	\begin{align*}
    		C_{D_x}(t) = \frac{\partial\rho_{D_x}(s,t)}{\partial s}\Bigr|_{s=t}, \ \ \ C_{W,r}(t) = \frac{\partial\rho_{W,r}(s,t)}{\partial s}\Bigr|_{s=t}
    	\end{align*}
    	with $\rho_{D_x}(u,v) = \cov\left(Z_{D_x}(u), Z_{D_x}(v)\right)$ and $\rho_{W,r}(u,v) = \cov\left(Z_{W,r}(u), Z_{W,r}(v)\right)$. The explicit expressions for $C_{D_x}(t)$ and $C_{W,r}(t)$ can be obtained through combinatorial analysis.
    \end{remark}

    \subsection{Checking Type I Error Control Under Finite n} \label{subsec:check}
            To assess the accuracy of the proposed $p$-value approximation under finite n, we compared the critical values obtained from the permutation procedure with those derived from the analytical approximation and examined the discrepancy between them. 'Ana' represent the analytical critical value  and 'Per' represents critical values obtained through 1000 permutation. Table \ref{Z_D:permute-ana-critical-value} shows the comparison for $Z_{D_1}, Z_{D_2}$ for gaussian and graphlet kernel respectively. We can see that the $p$-value approximations for both kernels are quite accurate. Table \ref{Z_W0.5:permute-ana-critical-value} and Table \ref{Z_W2:permute-ana-critical-value} shows the analytical and permutation critical values for $Z_W$ for $r=0.5$ and $r=2$ respectively. For $Z_W$, we see that the convergence to gaussian process could be slower, especially when $n_0$ is closer to the end, leading to a larger gap between analytical and permutation critical values.
    
      \begin{table}[h!]\
    	\centering 
        \small
        \setlength{\tabcolsep}{3pt}
    	\caption{Critical values for the single change-point scan statistic $\max_{n_0 \le t \le n_1} Z_{D_x}(t),x\in \{1,2\}$ at 0.05 significance level}
        \label{Z_D:permute-ana-critical-value}
        \begin{minipage}[t]{0.48\textwidth}
    	\begin{tabular}{ccccccc} 
    		\hline  
    		 \multicolumn{1}{c}{ER $p = 0.2 $} &\multicolumn{2}{c}{$n_0 = 100$}  &  \multicolumn{2}{c}{$n_0 = 25$} &  \multicolumn{2}{c}{$n_0 = 10$} \\ 
    		\hline
    		Kernel & Ana  &Per & Ana  &Per & Ana  &Per \\
    		\hline 
    		Gaussian  & & &&&&\\
    		$n = 500$ &2.81& 2.81 & 3.07 & 3.06  & 3.16 & 3.15\\ 
            \hline
    		Graphlet    &  &  &  &  &   \\
    		$n = 500$ &2.81& 2.80 & 3.07 & 3.08  & 3.16 & 3.15 \\
    		\hline 
    	\end{tabular} 
        \end{minipage}\hfill
        \begin{minipage}[t]{0.48\textwidth}
        \begin{tabular}{ccccccc} 
    		\hline  
            \multicolumn{1}{c}{$\quad\,$ SBM $\quad$} &\multicolumn{2}{c}{$n_0 = 100$}  &  \multicolumn{2}{c}{$n_0 = 25$} &  \multicolumn{2}{c}{$n_0 = 10$}   \\ 
    		\hline
    		Kernel & Ana  &Per & Ana  &Per & Ana  &Per  \\
    		\hline 
    		Gaussian  & & &&&&\\
    		$n = 500$ &2.82& 2.81 &3.08 & 3.06  & 3.16 & 3.12  \\ 
            \hline
    		Graphlet    &  &  &  &  &   \\
    		$n = 500$ &2.82& 2.84 &3.08 & 3.08  & 3.16 & 3.15 \\
    		\hline 
    	\end{tabular} 
        \end{minipage}
    \end{table} 
        \begin{table}[h]
        \centering
        \small
        \setlength{\tabcolsep}{3pt}
        \caption{Critical values for the single change-point scan statistic
        $\max_{n_0 \le t \le n_1} Z_{W,0.5}(t)$ at significance level $0.05$.}
        \label{Z_W0.5:permute-ana-critical-value}
        \begin{minipage}[t]{0.48\textwidth}
        \centering
        \begin{tabular}{ccccccc}
        \hline
        \multicolumn{1}{c}{ER, $p=0.2$}
        & \multicolumn{2}{c}{$n_0=100$}
        & \multicolumn{2}{c}{$n_0=25$}
        & \multicolumn{2}{c}{$n_0=10$}\\
        \hline
        Kernel & Ana & Per & Ana & Per & Ana & Per  \\
        \hline
        Gaussian &&&&&& \\
        $n=500$ & 2.53 & 2.57 & 2.83 & 2.87 & 2.92 & 2.96 \\
        \hline
        Graphlet &&&&&& \\
        $n=500$ & 2.52 & 2.58 & 2.82 & 2.89 & 2.91 & 3.00 \\
        \hline
        \end{tabular}
        \end{minipage}\hfill
        \begin{minipage}[t]{0.48\textwidth}
        \centering
        \begin{tabular}{ccccccc}
        \hline
        \multicolumn{1}{c}{SBM}
        & \multicolumn{2}{c}{$n_0=100$}
        & \multicolumn{2}{c}{$n_0=25$}
        & \multicolumn{2}{c}{$n_0=10$}\\
        \hline
        Kernel & Ana & Per & Ana & Per & Ana & Per \\
        \hline
        Gaussian &&&&&& \\
        $n=500$ & 2.56 & 2.62 & 2.87 & 2.89 & 2.97 & 2.98  \\
        \hline
        Graphlet &&&&&& \\
        $n=500$ & 2.52 & 2.65 & 2.82 & 2.98 & 2.91 & 3.11 \\
        \hline
        \end{tabular}
        \end{minipage}
        \end{table}
    \begin{table}[h!]
    	\centering 
        \small
        \setlength{\tabcolsep}{3pt}
    	\caption{Critical values for the single change-point scan statistic $\max_{n_0 \le t \le n_1} Z_{W,2}(t)$ at 0.05 significance level}
    	\label{Z_W2:permute-ana-critical-value}
        \begin{minipage}[t]{0.48\textwidth}
    	\begin{tabular}{ccccccc} 
    		\hline  
    		 \multicolumn{1}{c}{ER $p=0.2 $} &  \multicolumn{2}{c}{$n_0 = 100$} & \multicolumn{2}{c}{$n_0 = 25$} &  \multicolumn{2}{c}{$n_0 = 10$} \\ 
    		\hline
    		Kernel & Ana &Per  & Ana  &Per & Ana  &Per \\
    		\hline 
    		Gaussian  & & &&&\\
    		$n = 500$  &2.53& 2.58 & 2.83 & 2.88 & 2.92 & 2.94 \\ 
            \hline
    		Graphlet   &  &  &  &  &   \\
    		$n = 500$  &2.52&2.59 & 2.82 & 2.90 & 2.91 & 3.00 \\
    		\hline 
    	\end{tabular} 
        \end{minipage}
        \hfill
        \begin{minipage}[t]{0.48\textwidth}
        \begin{tabular}{ccccccc} 
    		\hline  
    		 \multicolumn{1}{c}{$\quad\,$ SBM $\quad$} &\multicolumn{2}{c}{$n_0 = 100$}  &  \multicolumn{2}{c}{$n_0 = 25$} &  \multicolumn{2}{c}{$n_0 = 10$} \\ 
    		\hline
    		Kernel & Ana  &Per & Ana  &Per & Ana  &Per  \\
    		\hline 
    		Gaussian  & & &&&&\\
    		$n = 500$ &2.56& 2.60 &2.87  & 2.90  & 2.97 & 3.03 \\ 
            \hline
    		Graphlet    &  &  &  &  &  \\
    		$n = 500$ &2.52& 2.65 & 2.82 & 2.97  & 2.91 & 3.13 \\
    		\hline 
    	\end{tabular} 
        \end{minipage}
    \end{table} 
    
    \section{Additional Experiments Details for Section \ref{sec:performance}} \label{appendix:experiments-details}

    We specify additional details for all methods implemented in Section \ref{sec:performance}. All experiments were run on a Linux CPU server with dual Intel Xeon E5-2699 v3 processors at $2.30$GHz.
    
    \textbf{Details on metrics: } For NBS and CPDstergm, the method is defaulted to produce multiple change-points. Under the single change-point alternative, in case where multiple change-points were produced, we classify the detection as accurate as long as there exists a $\hat \tau$ in the list of estimated change-points such that $|\hat \tau -\tau|\le 0.05n$.

    \textbf{Details on parameters: } For CPDstergm, we set the range of $\lambda$s to search from to be: $lambda = (10^{-2}, 10^4)$, and set the threshold alpha at 0.05. For NBS, we set the separation: $\Delta = 10$, and set the threshold at the default recommended level given at $N \hat \rho (\log(\tau))^2 / 20$, where $\hat \rho$ is sparsity parameter estimated by the 0.95 quantile of the estimated connectivity probability of each node. For all kernel-based methods, we set $K_1$ to be the gaussian kernel with median heuristic, and $K_2$ to be graphlet kernel with subgraph size of 3. $n_0,n_1$ at default values. Additionally for KAPf-CPD, we set $r_1=0.5$, $r_2=2$.
    
    \section{Additional Experiments Details for Section \ref{subsec:enron}}
        For binary-segmentation for KAP-CPD,GKCP, we set the minimum separation between two change-points at 6. For NBS, we set the separation: $\Delta = 6$, and set the threshold slightly higher than default level at $N \hat \rho (\log(\tau))^2 / 10$, where $\hat \rho$ is sparsity parameter estimated by the 0.95 quantile of the estimated connectivity probability of each node. For Fréchet, the original paper \cite{frechet-cpd} also analyzed the same dataset, so we directly reported their estimated change-points as documented in their paper.

    \section{Additional Experiments Details for Section \ref{subsec:seizure}}
        Since the networks in this dataset is weighted, and graphlet kernel only works on unweighted graphs. We did the following additional preprocessing to construct the graphlet kernel. We convert any weights $\ge 0.5$ or $\le -0.5$ to an edge, and the rest are not connected. Specific parameters setup is the same as in Appendix \ref{appendix:experiments-details}.

\end{document}